\documentclass[11pt,oneside]{article}

\usepackage[T1]{fontenc}
\usepackage[margin=1in]{geometry}
\usepackage{mathpazo}
\usepackage{xcolor}

\newcommand{\red}[1]{\textcolor{red}{#1}}

\usepackage[hidelinks]{hyperref}

\usepackage{booktabs}
\usepackage{caption}
\usepackage{graphicx}

\let\oldincludegraphics\includegraphics%
\renewcommand{\includegraphics}[2][]{\IfFileExists{#2}{\oldincludegraphics[#1]{#2}}{\red{[FILE NOT FOUND]}}}

\usepackage{amsmath}
\usepackage{amssymb}
\usepackage{amsthm}

\DeclareMathOperator{\E}{E}

\DeclareMathOperator{\Cov}{Cov}
\DeclareMathOperator{\Var}{Var}

\newcommand{\Abs}[1]{\left\lvert#1\right\rvert}
\newcommand{\abs}[1]{\lvert#1\rvert}
\newcommand{\der}{\mathrm{d}}

\newcommand{\Hcal}{\mathcal{H}}
\newcommand{\N}{\mathbb{N}}
\newcommand{\Ncal}{\mathcal{N}}
\newcommand{\parfrac}[2]{\frac{\partial #1}{\partial #2}}
\newcommand{\R}{\mathbb{R}}
\renewcommand{\epsilon}{\varepsilon}

\let\oldleft\left
\let\oldright\right
\renewcommand{\left}{\mathopen{}\mathclose\bgroup\oldleft}
\renewcommand{\right}{\aftergroup\egroup\oldright}

\newtheorem{lemma}{Lemma}
\newtheorem{proposition}{Proposition}

\newif\ifbodyproofs
\bodyproofstrue

\usepackage{needspace}
\AtBeginEnvironment{lemma}{\Needspace*{8\baselineskip}}
\AtBeginEnvironment{proposition}{\Needspace*{8\baselineskip}}
\AtBeginEnvironment{theorem}{\Needspace*{8\baselineskip}}

\usepackage{natbib}
\bodyproofsfalse

\title{Learning about a changing state}
\author{%
Benjamin Davies\thanks{
Department of Economics, Stanford University; bldavies@stanford.edu.
I thank Arun Chandrasekhar, Ben Golub, Matt Jackson, and Anirudh Sankar for helpful discussions.
}
}
\date{Draft version: \today}

\begin{document}

\maketitle
\thispagestyle{empty}

\begin{abstract}
    \noindent
    A long-lived Bayesian agent observes costly signals of a time-varying state.
    He chooses the signals' precisions sequentially, balancing their costs and marginal informativeness.
    I compare the optimal myopic and forward-looking precisions when the state follows a Brownian motion.
    I also compare the myopic precisions induced by other Gaussian processes.

    \vskip\baselineskip
\end{abstract}

\clearpage
\setcounter{page}{1}
\section{Introduction}

This paper studies a tractable model of dynamic information acquisition.
I consider a long-lived Bayesian agent who buys independent signals of an unknown state.
In contrast to other models that assume the state is fixed \citep[see, e.g.,][]{Che-Mierendorff-2019-AER,Fudenberg-etal-2018-AER,Hebert-Woodford-2023-JET,Liang-etal-2018-,Liang-etal-2022-ECTA,Moscarini-Smith-2001-ECTA}, I allow it to vary over time.
This aligns my model with many real-world settings:
\begin{itemize}

    \item
    \emph{News consumption}:
    How much attention should we pay to current events, given our knowledge of past events and of news cycle lengths?


    \item
    \emph{Drug development}:
    How much should laboratories invest in testing new vaccines, given their tests of previous vaccines and their beliefs about how viruses mutate over time?

    \item
    \emph{Machine learning}:
    How does ``concept drift,'' wherein the distribution of a target parameter changes over time, affect how much training data is needed and how often it is updated?

    \item
    \emph{Intergenerational advice}:
    How much can we learn from others' experiences (e.g., from parents and mentors), given the difference between their context and ours?

\end{itemize}

Section~\ref{sec:model} describes my model and discusses its assumptions.
The state follows the sample path of a Brownian motion with known drift and scale, but unknown initial value and increments.
The agent learns about the state by observing signals of its value at countably many times.
He uses these signals to form posterior beliefs and take actions with quadratic state-dependent costs.
The agent chooses the actions and signal precisions that minimize the sum of his expected action and information costs.
He makes these choices sequentially and myopically, based on the full history of previous signals but without regard for his future payoffs.

Section~\ref{sec:solution} analyzes the agent's optimal myopic precisions.
Proposition~\ref{prop:solution-brownian} states that his optimal learning strategy has two stages: he waits until information is more beneficial than costly, then buys information at every opportunity to maintain his target posterior variance.

Sections~\ref{sec:planning} and~\ref{sec:other-processes} extend my model in different ways.
Section~\ref{sec:planning} augments the agent's objective to include the present value of his future payoffs.
This preserves his optimal two-stage learning strategy, but changes the time he waits and precisions he chooses.
Intuitively, the more the agent cares about his future payoffs, the more information he buys because it allows him to buy less in the future while maintaining his target posterior variance.

Section~\ref{sec:other-processes} considers other state-generating processes.
I explain how to compute the optimal myopic precisions when the state follows \emph{any} Gaussian process.
I consider two such processes: the Ornstein-Uhlenbeck process and a linear process.
The former induces similar learning dynamics to the Brownian motion: the agent may never buy information, but if he starts then he never stops.
In contrast, if the state is linear in time, and the time between signals is constant, then the agent always starts buying information eventually.
However, he may also stop buying it eventually as he learns the state's initial value and growth rate.

Section~\ref{sec:literature} discusses related literature.
Section~\ref{sec:conclusion} concludes.
Appendix~\ref{sec:proofs} contains proofs of my mathematical claims.

\section{Model}
\label{sec:model}

A Bayesian agent receives signals~$s_1,s_2,\ldots$ of a time-varying state~$\{\theta(t)\}_{t\ge0}$.
He uses these signals to form posterior beliefs and take actions with state-dependent costs.
The agent chooses the signal precisions and actions sequentially to minimize their expected costs.

This section describes the agent's learning and choice environment, and discusses my modeling assumptions.

\subsection{The agent's environment}

\paragraph{Evolution of~$\{\theta(t)\}_{t\ge0}$}

The state follows a Brownian motion with drift~$\mu\in\R$, scale~$\sigma>0$, and initial value~$\theta_0\equiv\theta(0)$.
It solves the stochastic differential equation (SDE)
\begin{equation}
    \der \theta(t)=\mu\der t+\sigma\der W(t), \label{eq:sde-brownian}
\end{equation}
where~$\{W(t)\}_{t\ge0}$ is the sample path of a Wiener process.
This path has initial value~$W(0)=0$, and iid normally distributed increments~$\der W(t)\equiv W(t+\der t)-W(t)$ with mean zero and variance~$\der t$:
\[ \der W(t)\sim\Ncal(0,\der t). \]
The agent knows~$\mu$ and~$\sigma$, but does not know~$\theta_0$ or~$\{W(t)\}_{t\ge0}$.
His prior specifies~$\theta_0\sim\Ncal(0,\sigma_0^2)$ independently of~$\{W(t)\}_{t\ge0}$.
Thus
\[ \theta(t)\sim\Ncal(\mu t,\sigma_0^2+\sigma^2t) \]
and
\begin{equation}
    \Cov(\theta(t),\theta(t'))=\sigma_0^2+\sigma^2\min\{t,t'\} \label{eq:covariance-brownian}
\end{equation}
for all times~$t\ge0$ and~$t'\ge0$.

\paragraph{Learning about~$\{\theta(t)\}_{t\ge0}$}

Each signal
\[ s_n=\theta(t_n)+\epsilon_n \]
estimates the state's value~$\theta(t_n)$ at a time~$t_n\ge0$ that is strictly increasing in~$n\in\N=\{1,2,\ldots\}$.
The error
\[ \epsilon_n\sim\Ncal\left(0,\frac{1}{p_n}\right) \]
in this signal has mean zero and precision~$p_n\equiv1/\Var(\epsilon_n)$, and is independent of~$\{\theta(t)\}_{t\ge0}$ and the errors in other signals.
The agent uses the history~$\Hcal_n\equiv\{s_i\}_{i=1}^n$ of signals received by time~$t_n$ to form posterior beliefs about the realized path~$\{\theta(t)\}_{t\ge0}$ via Bayes' rule.

\begin{figure}[t]
    \centering
    \includegraphics[width=0.75\linewidth]{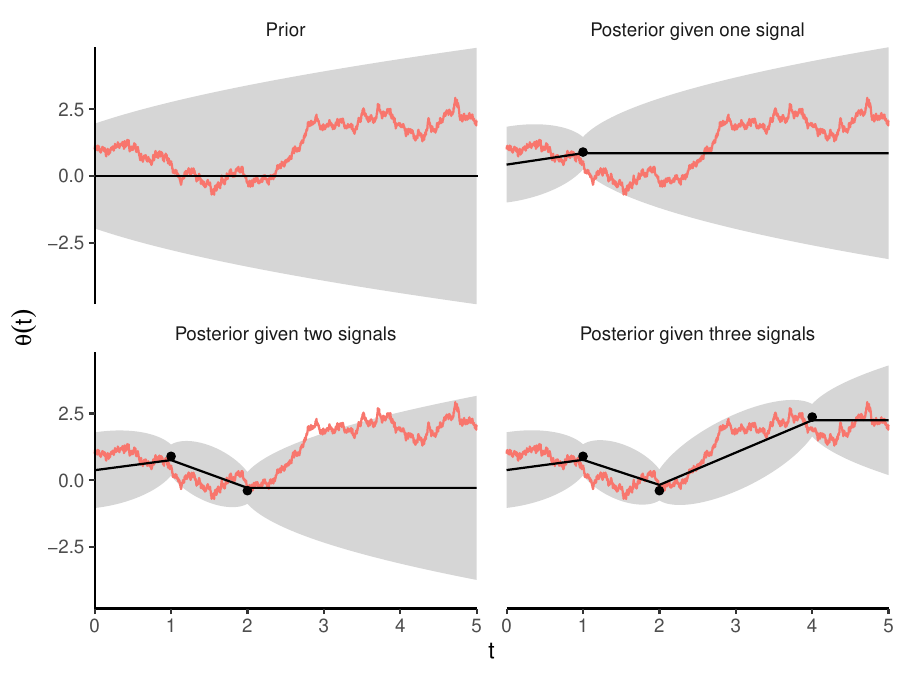}
    \caption{Learning about~$\{\theta(t)\}_{t\ge0}$}
    \label{fig:learning}
    \caption*{
    \footnotesize
    {\itshape Notes:}
    Red lines represent realization of~$\{\theta(t)\}_{t\ge0}$ with~$(\mu,\sigma,\sigma_0)=(0,1,1)$.
    Black dots represent signals~$s_1,s_2,s_3$ with precisions~$p_n=10$.
    Black lines represent prior and posterior means.
    Gray regions represent 95\% confidence intervals, constructed analytically from the prior and posterior means and variances defined in Lemma~\ref{lem:posterior}.
    }
\end{figure}

I illustrate the agent's learning process in Figure~\ref{fig:learning}.
It shows his prior and posterior means after receiving one, two, and three signals.
It also shows the 95\% confidence intervals around these means.
These intervals narrow as~$t$ approaches the times at which the agent receives signals.
The posterior mean is piecewise linear in~$t$ between signals and constant in~$t$ after the last signal.
This is because the random increments~$\der\theta(t)$ are independent, which makes each signal~$s_n$ informative about~$\{\theta(t)\}_{t\ge0}$ only insofar as it is informative about the realized path between~$(0,\theta_0)$ and~$(t_n,\theta(t_n))$.

\paragraph{Choices and payoffs}

At each time~$t_n$, the agent chooses the precision~$p_n\ge0$ and action~$a_n\in\R$ that jointly minimize
\[ \E[(a_n-\theta(t_n))^2\mid\Hcal_n]+cp_n, \]
where~$c>0$ is the marginal cost of information.%
\footnote{
The cost function~$cp_n$ maps to a setting where the agent can buy any non-negative real number of iid signals at a constant per-unit price.
For example, if each signal has precision~$1/\sigma_\epsilon^2\ge0$ and price~$\kappa>0$, then~$m\ge0$ signals have combined precision~$p\equiv m/\sigma_\epsilon^2$ and cost~$\kappa m=cp$ with~$c\equiv\kappa\sigma_\epsilon^2$.
}
His optimal action~$a_n^*=\E[\theta(t_n)\mid\Hcal_n]$ equals the posterior mean of~$\theta(t_n)$ given~$\Hcal_n$.
Therefore, his optimal precision~$p_n^*$ minimizes%
\footnote{
Alternatively, the agent could minimize~$\Var(\theta(t_n)\mid\Hcal_n)$ subject to a budget constraint~$cp_n\le B$ for some~$B>0$ \citep[as in, e.g.,][]{Immorlica-etal-2021-LeibnizInt.Proc.Inform.LIPIcs}.
Then his optimal precision~($B/c$) would depend on the cost of new information but not the value of such information.
I highlight this tension between cost and value by including~$cp_n$ in the agent's objective.
}
\begin{align*}
    \psi_n(p_n)
    &\equiv \E\left[\left(\E[\theta(t_n)\mid\Hcal_n]-\theta(t_n)\right)^2\mid\Hcal_n\right]+cp_n \\
    &= \Var(\theta(t_n)\mid\Hcal_n)+cp_n.
\end{align*}
%
Increasing~$p_n$ makes the signal~$s_n$ of~$\theta(t_n)$ more expensive but more informative.%
\footnote{
Thus~$c$ indexes the cost of information relative to the value of knowing~$\theta(t_n)$ precisely.
}
The value of this information (i.e., the variance reduction it delivers) depends on previous signals' precisions and on the state's variation over time.
This variation makes signals gradually lose explanatory power, making~$\Var(\theta(t_n)\mid\Hcal_n)$ more sensitive to recent signals' precisions.

\subsection{Discussion of modeling assumptions}

\paragraph{Gaussian states and signals}

I assume states~$\{\theta(t_n)\}_{n\in\N}$ and signals~$\{s_n\}_{n\in\N}$ are jointly normally distributed with known means and variances.
Many related studies also make this assumption.
It allows me to derive closed-form expressions for the agent's posterior beliefs.
It also makes posterior variances independent of signals' realizations.
This means the agent knows the \emph{ex post} value of the information he will receive before he receives it.%
\footnote{
This would not be true if, e.g., states and signals were binary.
Then the states' posterior variances would depend on signals' realizations.
}
This knowledge is especially important when he internalizes how his choice of~$p_n$ affects his choices~$p_{n+1},p_{n+2},\ldots$ (see Section~\ref{sec:planning}).

Assuming~$\{\theta(t)\}_{t\ge0}$ follows a Brownian motion allows me to derive closed-form expressions for the optimal precisions~$p_n^*$.
It captures how the state changes over time in response to a continuous stream of unpredictable ``shocks'' (e.g., news events or virus mutations).
Section~\ref{sec:other-processes} discusses the optimal precisions induced by other state-generating processes.

\paragraph{Quadratic payoffs}

I assume actions~$a_n$ have quadratic costs~$(a_n-\theta(t_n))^2$.
This makes optimal actions~$a_n^*=\E[\theta(t_n)\mid\Hcal_n]$ have expected costs~$\E[(a_n^*-\theta(t_n))^2\mid\Hcal_n]=\Var(\theta(t_n)\mid\Hcal_n)$ equal to posterior variances.
Many related studies also make this assumption.
It makes my model tractable and captures the intuition that information is more valuable when it yields more precise beliefs.
\cite{Frankel-Kamenica-2019-AER} formalize this intuition and extend it to other payoff structures.

\paragraph{Linear costs}

I assume signals' costs are linear in their precisions.
\cite{Pomatto-etal-2023-AER} show that linear cost functions uniquely satisfy some attractive properties.
Such functions appear in many sequential sampling models \citep[e.g.,][]{Morris-Strack-2019-,Wald-1945-AnnMathStat,Wald-1947-ECTA} and their continuous-time analogues \citep[e.g.,][]{Fudenberg-etal-2018-AER,Liang-etal-2022-ECTA}.

\paragraph{Countable signals}

I assume the agent receives countably many signals.
This allows me to leverage results about finite collections of Gaussian variables (e.g., Lemmas~\ref{lem:multivariate-normal-conditional} and~\ref{lem:posterior}).
However, my analysis only requires the time steps~$\Delta t_n\equiv t_n-t_{n-1}$ to be strictly positive.
They could be arbitrarily small, allowing my model to approximate a setting where the agent receives information continuously.
If the~$\Delta t_n$ are constant in~$n$, then the sequence~$\{\theta(t_n)\}_{n\in\N}$ is a discrete-time random walk with Gaussian shocks (as studied, e.g., by \cite{Dasaratha-etal-2023-REStud} and \cite{Immorlica-etal-2021-LeibnizInt.Proc.Inform.LIPIcs}).

\paragraph{Myopic preferences}

I assume the agent is myopic in that his optimal precision~$p_n^*$ depends on the previous precisions~$p_1^*,p_2^*,\ldots,p_{n-1}^*$ but not the future precisions~$p_{n+1}^*,p_{n+2}^*,\ldots$.
This makes my model tractable, and captures a setting where the agent
(i)~does not know the times at which he will receive information in the future,
(ii)~has cognitive limitations that prevent him from ``planning ahead,'' or
(iii)~does not care about his future payoffs.
Section~\ref{sec:planning} compares the optimal precisions when the agent is myopic to those when he minimizes the present value of his current and future payoffs.
The myopic precisions are smaller because they ignore the future benefit of being more informed.
However, the agent's optimal learning strategy does not differ qualitatively between the myopic and forward-looking cases.

\section{Optimal myopic precisions}
\label{sec:solution}

At each time~$t_n$, the agent chooses the precision~$p_n^*$ that solves
\begin{equation}
    \min_{p_n}\psi_n(p_n)\ \text{subject to}\ p_n\ge0 \label{eq:problem}
\end{equation}
given his previous choices~$p_1^*,p_2^*,\ldots,p_{n-1}^*$.
Lemma~\ref{lem:solution} characterizes~$p_n^*$ and the resulting posterior variance~$\Var(\theta(t_n)\mid\Hcal_n)$.

\begin{lemma}
    \label{lem:solution}
    Let~$n\in\N$ and let
    \[ R_n\equiv \Var(\theta(t_n)\mid\Hcal_n\setminus\{s_n\}) \]
    denote the variance of~$\theta(t_n)$ that is not explained by the signals in~$\Hcal_n\setminus\{s_n\}$.
    Then the unique solution
    \begin{equation}
        p_n^*=\max\left\{\frac{1}{\sqrt{c}}-\frac{1}{R_n},0\right\} \label{eq:optimal-precision}
    \end{equation}
    to~\eqref{eq:problem} yields posterior variance
    \[ \Var(\theta(t_n)\mid\Hcal_n)=\min\{R_n,\sqrt{c}\}. \]
\end{lemma}
\ifbodyproofs\begin{proof}[Proof of Lemma~\ref{lem:solution}]
    Let~$z_n\equiv\theta(t_n)\mid\Hcal_n\setminus\{s_n\}$ and~$s_n'\equiv s_n\mid\Hcal_n\setminus\{s_n\}$ be the components of~$\theta(t_n)$ and~$s_n$ that are not explained by the signals in~$\Hcal_n\setminus\{s_n\}$.
    Then~$\Var(z_n)=R_n$ by definition, while
    \begin{align*}
        \Cov(z_n,s_n')
        &= \Cov(\theta(t_n),s_n\mid\Hcal_n\setminus\{s_n\}) \\
        &= \Cov(\theta(t_n),\theta(t_n)+\epsilon_n\mid\Hcal_n\setminus\{s_n\}) \\
        &= \Var(\theta(t_n)\mid\Hcal_n\setminus\{s_n\}) \\
        &= R_n
    \end{align*}
    and
    \begin{align*}
        \Var(s_n')
        &= \Var(\theta(t_n)+\epsilon_n\mid\Hcal_n\setminus\{s_n\}) \\
        &= \Var(\theta(t_n)\mid\Hcal_n\setminus\{s_n\})+\Var(\epsilon_n) \\
        &= R_n+\frac{1}{p_n}
    \end{align*}
    because the error~$\epsilon_n$ is independent of~$\theta(t_n)$ and~$\Hcal_n\setminus\{s_n\}$.
    So, by Lemma~\ref{lem:multivariate-normal-conditional}, we have
    \begin{align*}
        \Var(\theta(t_n)\mid\Hcal_n)
        &= \Var(z_n\mid s_n') \\
        &= \Var(z_n)-\frac{\left(\Cov(z_n,s_n')\right)^2}{\Var(s_n')} \\
        &= R_n-\frac{R_n^2}{R_n+1/p_n} \\
        &= \frac{R_n}{1+p_nR_n}
    \end{align*}
    and therefore
    \[ \psi_n(p_n)=\frac{R_n}{1+p_nR_n}+cp_n. \]
    Differentiating twice with respect to~$p_n$ gives
    \[ \psi_n'(p_n)=-\left(\frac{R_n}{1+p_nR_n}\right)^2+c \]
    and
    \[ \psi_n''(p_n)=2\left(\frac{R_n}{1+p_nR_n}\right)^3. \]
    If~$R_n=0$ then~$\psi_n'(p_n)>0$ for all~$p_n\in\R$, and so~$\psi_n(p_n)$ attains its constrained minimum at~$p_n=0$.
    So suppose~$R_n>0$.
    Then~$\psi_n(p_n)$ is strictly convex in~$p_n\ge0$ and its unique constrained minimizer equals zero precisely when
    \begin{align*}
        0
        &\le \psi_n'(0) \\
        &= c-R_n^2,
    \end{align*}
    which holds if and only if~$\sqrt{c}\ge R_n$.
    Otherwise, this minimizer is strictly positive and satisfies the first-order condition~$\psi'(p_n)=0$.
    Solving this equation for~$p_n$ yields
    \[ p_n=\frac{1}{\sqrt{c}}-\frac{1}{R_n}. \]
    Equation~\eqref{eq:optimal-precision} follows.
    Moreover,
    \[ \Var(\theta(t_n)\mid\Hcal_n)\big\rvert_{p_n=0}=R_n \]
    and
    \begin{align*}
        \Var(\theta(t_n)\mid\Hcal_n)\big\rvert_{p_n=\frac{1}{\sqrt{c}}-\frac{1}{R_n}}
        &= \frac{R_n}{1+\left(\frac{1}{\sqrt{c}}-\frac{1}{R_n}\right)R_n} \\
        &= \sqrt{c}.\qedhere
    \end{align*}
\end{proof}
\fi

The ``residual variance''~$R_n$ measures the agent's knowledge of~$\theta(t_n)$ before observing~$s_n$.
This variance is smaller when the agent has more prior information.
Indeed, his posterior precision
\[ \frac{1}{\Var(\theta(t_n)\mid\Hcal_n)}=\frac{1}{R_n}+p_n^* \]
equals the sum of his prior (to receiving~$s_n$) and signal precisions.
If~$R_n$ is smaller than~$\sqrt{c}$, then the agent does not buy more information about~$\theta(t_n)$ because its cost outweighs its benefit.
Then
\[ \Var(\theta(t_n)\mid\Hcal_n)\big\rvert_{p_n^*=0}=R_n \]
because the signal~$s_n$ is uninformative.
If~$R_n$ is larger than~$\sqrt{c}$, then the agent buys more information and
\[ \Var(\theta(t_n)\mid\Hcal_n)\big\rvert_{p_n^*>0}=\sqrt{c}. \]
Increasing~$R_n$ or decreasing~$c$ makes new information more net valuable, increasing the optimal precision~$p_n^*$ and decreasing the resulting posterior variance~$\Var(\theta(t_n)\mid\Hcal_n)$.

Lemma~\ref{lem:solution} does not rely on~$\{\theta(t)\}_{t\ge0}$ following a Brownian motion.%
\footnote{
Neither does Lemma~\ref{lem:solution} rely on~$\Hcal_n$ equaling~$\{s_i\}_{i=1}^n$.
It holds whenever~$\Hcal_n$ is a subset of~$\{s_i\}_{i\in\N}$ containing~$s_n$.
}
It holds whenever the states~$\theta(t_1),\theta(t_2),\ldots,\theta(t_n)$ and signals~$s_1,s_2,\ldots,s_n$ are jointly normally distributed.
This property obtains when~$\{\theta(t)\}_{t\ge0}$ follows a Gaussian process.
The Brownian motion is one example of such a process.
I discuss other examples in Section~\ref{sec:other-processes}.

The Brownian motion is also a Markov process: its future values are conditionally independent of its past values given its current value.
Lemma~\ref{lem:residual-variance-brownian} uses this Markov property to characterize the residual variance~$R_n$ of~$\theta(t_n)$ given the signals in~$\Hcal_n\setminus\{s_n\}$.

\begin{lemma}
    \label{lem:residual-variance-brownian}
    Let~$n\in\N$.
    Suppose~$\{\theta(t)\}_{t\ge0}$ follows a Brownian motion with scale~$\sigma>0$ and initial value~$\theta_0\sim\Ncal(0,\sigma_0^2)$.
    Then
    \begin{equation}
        R_n=\begin{cases}
            \sigma_0^2+\sigma^2t_1 & \text{if}\ n=1 \\
            \Var(\theta(t_{n-1})\mid\Hcal_{n-1})+\sigma^2(t_n-t_{n-1}) & \text{if}\ n>1.
        \end{cases} \label{eq:residual-variance-brownian}
    \end{equation}
\end{lemma}
\ifbodyproofs\input{proofs/residual-variance-bronwian}\fi

The residual variance~$R_1$ equals the prior variance of~$\theta(t_1)$ because~$\Hcal_1\setminus\{s_1\}$ is empty.
For each~$n>1$, equation~\eqref{eq:residual-variance-brownian} decomposes~$R_n$ as the sum of two terms:
the posterior variance of~$\theta(t_{n-1})$ given signals~$s_1$ through~$s_{n-1}$, and
the variance of the integral
\[ \int_{t_{n-1}}^{t_n}\der\theta(t)=\mu(t_n-t_{n-1})+\sigma\int_{t_{n-1}}^{t_n}\der W(t) \]
of the random increments between times~$t_{n-1}$ and~$t_n$.
Intuitively, the residual variance~$R_n$ measures how accurately the agent can forecast~$\theta(t_n)$ based on his knowledge of~$\theta(t_{n-1})$.
His forecast is more accurate when~$\Var(\theta(t_{n-1})\mid\Hcal_{n-1})$ and~$\sigma^2(t_n-t_{n-1})$ are smaller, since then his knowledge of the state is more precise and ``up-to-date.''

Combining Lemmas~\ref{lem:solution} and~\ref{lem:residual-variance-brownian} yields a closed-form expression for the optimal precision~$p_n^*$:

\begin{proposition}
    \label{prop:solution-brownian}
    Let~$n\in\N$.
    Suppose~$\{\theta(t)\}_{t\ge0}$ follows a Brownian motion with scale~$\sigma>0$ and initial value~$\theta_0\sim\Ncal(0,\sigma_0^2)$.
    Define
    \[ \overline{t}\equiv\frac{\sqrt{c}-\sigma_0^2}{\sigma^2}. \]
    If~$p_i=p_i^*$ for each~$i<n$, then
    \begin{equation}
        p_n^*=\begin{cases}
            0 & \text{if}\ t_n\le\overline{t} \\
            \frac{1}{\sqrt{c}}-\frac{1}{\sigma_0^2+\sigma^2t_1} & \text{if}\ n=1\ \text{and}\ \overline{t}<t_1 \\
            \frac{1}{\sqrt{c}}-\frac{1}{\sigma_0^2+\sigma^2t_n} & \text{if}\ n>1\ \text{and}\ \overline{t}<t_n\le\overline{t}+(t_n-t_{n-1}) \\
            \frac{1}{\sqrt{c}}-\frac{1}{\sqrt{c}+\sigma^2(t_n-t_{n-1})} & \text{if}\ n>1\ \text{and}\ \overline{t}+(t_n-t_{n-1})<t_n
        \end{cases}
    \end{equation}
    and
    \[ \Var(\theta(t_n)\mid\Hcal_n)=\min\{\sigma_0^2+\sigma^2t_n,\sqrt{c}\}. \]
\end{proposition}
\ifbodyproofs\begin{proof}[Proof of Proposition~\ref{prop:solution-brownian}]
    Suppose~$n>1$.
    The Brownian motion is a Gauss-Markov process, so we can use Lemma~\ref{lem:residual-variance-markov} to characterize~$R_n$.
    Define~$\beta_n$ and~$\gamma_n$ as in that lemma.
    Then~$\beta_n=1$ and~$\gamma_n=\sigma^2(t_n-t_{n-1})$ from the proof of Lemma~\ref{lem:residual-variance-brownian}.
    So
    \[ \prod_{j=0}^{n-2}\beta_{n-j}^2=1 \]
    and
    \begin{align*}
        \sum_{j=0}^{n-2}\gamma_{n-j}\prod_{k=0}^{j-1}\beta_{n-k}^2
        &= \sum_{j=0}^{n-2}\sigma^2(t_{n-j}-t_{n-j-1}) \\
        &= \sigma^2(t_n-t_1),
    \end{align*}
    and similarly
    \[ \prod_{j=0}^{i-1}\beta_{n-j}^2=1 \]
    and
    \[ \sum_{j=0}^{i-1}\gamma_{n-j}\prod_{k=0}^{j-1}\beta_{n-k}^2=\sigma^2(t_n-t_{n-i}) \]
    for each~~$i<n$.
    Thus, if~$p_i=p_i^*$ for each~$i<n$, then
    \begin{align*}
        R_n
        &= \min\left\{\Var(\theta(t_1))\prod_{j=0}^{n-2}\beta_{n-j}^2+\sum_{j=0}^{n-2}\gamma_{n-j}\prod_{k=0}^{j-1}\beta_{n-k}^2\right\}\cup\left\{\sqrt{c}\prod_{j=0}^{i-1}\beta_{n-j}^2+\sum_{j=0}^{i-1}\gamma_{n-j}\prod_{k=0}^{j-1}\beta_{n-k}^2\right\}_{i=1}^{n-1} \\
        &= \min\left\{\left(\sigma_0^2+\sigma^2t_1\right)+\sigma^2(t_n-t_1)\right\}\cup\left\{\sqrt{c}+\sigma^2(t_n-t_{n-i})\right\}_{i=1}^{n-1} \\
        &= \min\left\{\sigma_0^2+\sigma^2t_n,\sqrt{c}+\sigma^2(t_n-t_{n-1})\right\}
    \end{align*}
    by Lemma~\ref{lem:residual-variance-markov}.
    The result follows from Lemma~\ref{lem:solution} and the definition of~$\overline{t}$.
\end{proof}
\fi

The agent's optimal learning strategy is as follows:
He does not buy any information until~$t_n$ exceeds~$\overline{t}$.
Then he buys informative signals to maintain his target posterior variance~$\sqrt{c}$.
These signals' precisions optimally balance the cost of information with the benefit of explaining some of the variance of the state that accumulates between times~$t_{n-1}$ and~$t_n$.
This benefit is large enough to ensure~$p_n^*>0$ for all times~$t_n>\overline{t}$.
Thus, the agent may wait arbitrarily long to start buying information, but once he starts he never stops.

If the timesteps~$t_n-t_{n-1}=\Delta t>0$ are constant in~$n$, then the optimal ``steady-state'' precision
\[ \lim_{n\to\infty}p_n^*=\frac{1}{\sqrt{c}}-\frac{1}{\sqrt{c}+\sigma^2\Delta t} \]
and payoff
\[ \lim_{n\to\infty}\psi_n(p_n^*)=2\sqrt{c}-\frac{c}{\sqrt{c}+\sigma^2\Delta t} \]
are both increasing in~$\sigma^2\Delta t$.
Intuitively, if the state accumulates more variance between consecutive times, then the agent has to buy more information to obtain his target posterior variance~$\sqrt{c}$.

\section{Optimal forward-looking precisions}
\label{sec:planning}

This section considers the case when the agent cares about his future payoffs.
At each time~$t_n$, he chooses the precisions~$p_n,p_{n+1},\ldots\ge0$ that minimize the present value
\begin{equation}
    \Psi_n((p_i)_{i\ge n})\equiv \sum_{i=n}^\infty\delta^{t_i-t_n}\left(\Var(\theta(t_i)\mid\Hcal_i)+cp_i\right) \label{eq:planning-present-value}
\end{equation}
of his time~$t_n,t_{n+1},\ldots$ payoffs, where~$\delta\in(0,1)$ is his subjective discount factor.
For simplicity, I assume~$t_{n+1}=t_n+1$ for each~$n\in\N$.
Then~\eqref{eq:planning-present-value} depends on~$n$ only via the residual variance~$R_n$.%
\footnote{
\label{fn:planning-present-value}%
This is because
\begin{align*}
    R_{i+1}
    &= \Var(\theta(t_i)\mid\Hcal_i)+\sigma^2 \\
    &= \left(\frac{1}{R_i}+p_i\right)^{-1}+\sigma^2
\end{align*}
for each~$i\ge n$ by Lemma~\ref{lem:residual-variance-brownian} and the definition of~$R_i$.
So the dynamics of~$(R_i)_{i\ge n}$ and~$(p_i)_{i\ge n}$ only depend on their initial values, which only depend on~$n$ via~$R_n$.
}%
\textsuperscript{,}\footnote{
Letting~$t_{n+1}-t_n=\Delta t\not=1$ for each~$n\in\N$, and replacing~$\sigma^2$ with~$\sigma^2/\Delta t$ and~$\delta$ with~$\delta^{\Delta t}$, yields the same insights.
}
This allows me to solve the agent's problem
\begin{equation}
    \min_{(p_i)_{i\ge n}}\Psi_n\left((p_i)_{i\ge n}\right)\ \text{subject to}\ p_i\ge0\ \text{for each}\ i\ge n \label{eq:planning-problem}
\end{equation}
by modeling it as a Markov decision process \citep{Bellman-1957-JMM}, where the residual variances~$(R_i)_{i\ge n}$ serve as state variables and the precisions~$(p_i)_{i\ge n}$ as control variables.

Proposition~\ref{prop:planning-solution} characterizes the solution to~\eqref{eq:planning-problem}.
It implies a similar optimal learning strategy to when the agent is myopic: he waits until information is more beneficial than costly, then buys more information at every opportunity.
But the time he waits to start buying and the amount he buys once he starts depend on his discount factor~$\delta$.
Intuitively, the more the agent cares about his future payoffs, the more he is willing to pay for information because the more he benefits from being informed in the future.

\begin{proposition}
    \label{prop:planning-solution}
    Suppose~$\{\theta(t)\}_{t\ge0}$ follows a Brownian motion with scale~$\sigma>0$ and initial value~$\theta_0\sim\Ncal(0,\sigma_0^2)$.
    Let~$t_{n+1}=t_n+1$ for each~$n\in\N$, let~$V>0$ solve
    \begin{equation}
        \frac{1}{c}=\frac{1}{V^2}-\frac{\delta}{(V+\sigma^2)^2}, \label{eq:planning-V}
    \end{equation}
    and define~$\overline{t}\equiv(V-\sigma_0^2)/\sigma^2$.
    Then the sequence~$(p_n^\dag)_{n\in\N}$ defined by
    \[ p_n^\dag=\begin{cases}
        0 & \text{if}\ t_n\le\overline{t} \\
        \frac{1}{V}-\frac{1}{\sigma_0^2+\sigma^2t_1} & \text{if}\ n=1\ \text{and}\ \overline{t}<t_1 \\
        \frac{1}{V}-\frac{1}{\sigma_0^2+\sigma^2t_n} & \text{if}\ n>1\ \text{and}\ \overline{t}<t_n\le\overline{t}+1 \\
        \frac{1}{V}-\frac{1}{V+\sigma^2} & \text{if}\ n>1\ \text{and}\ \overline{t}+1<t_n
    \end{cases} \]
    for each~$n\in\N$ solves~\eqref{eq:planning-problem}.
    Moreover, if~$p_i=p_i^\dag$ for each~$i\le n$ then
    \[ \Var(\theta(t_n)\mid\Hcal_n)=\min\left\{\sigma_0^2+\sigma^2t_n,V\right\}. \]
\end{proposition}
\ifbodyproofs\begin{proof}[Proof of Proposition~\ref{prop:planning-solution}]
    As stated in the text, the present value~\eqref{eq:planning-present-value} depends on~$n$ only via the residual variance~$R_n$ (see Footnote~\ref{fn:planning-present-value}).
    Let~$\Psi(R_n)$ denote this present value when the agent chooses the precisions~$p_n,p_{n+1},\ldots$ that solve~\eqref{eq:planning-problem}.
    Then
    \begin{align}
        \Psi(R_n)
        &= \min_{p_n\ge0}\left\{\left(\frac{1}{R_n}+p_n\right)^{-1}+cp_n+\delta\Psi\left(\left(\frac{1}{R_n}+p_n\right)^{-1}+\sigma^2\right)\right\} \notag \\
        &= \min_{V_n\in(0,R_n]}\left\{V_n+c\left(\frac{1}{V_n}-\frac{1}{R_n}\right)+\delta\Psi\left(V_n+\sigma^2\right)\right\}, \label{eq:planning-solution-bellman}
    \end{align}
    where the posterior variance
    \begin{align*}
        V_n
        &\equiv \Var(\theta(t_n)\mid\Hcal_n) \\
        &= \left(\frac{1}{R_n}+p_n\right)^{-1}
    \end{align*}
    is an invertible function of~$p_n$ given~$R_n$ and so the change of variables is well-defined.
    We want to find a function~$\Psi:(0,\infty)\to(0,\infty)$ satisfying the Bellman equation~\eqref{eq:planning-solution-bellman}.
    This is equivalent to finding a function~$\Phi:(0,\infty)\to(0,\infty)$ mapping residual variances to optimal posterior variances.
    
    Assume~$\Psi$ and~$\Phi$ are continuous and piecewise differentiable.
    (One can verify these properties using arguments provided in Section~3.3 of \citet{Stokey-Lucas-1989-}.)

    Now~$V_n\le R_n$ with equality if and only if~$p_n=0$.
    But choosing~$p_n=0$ is optimal if and only if the cost of buying information exceeds its value, which occurs precisely when~$R_n$ is sufficiently small.
    So there exists~$R>0$ such that~$\Phi(R_n)=R_n$ for all~$R_n<R$ and~$\Phi(R_n)<R_n$ for all~$R_n>R$.

    Suppose~$R_n>R$.
    Then the constraint~$V_n\le R_n$ on the RHS of~\eqref{eq:planning-solution-bellman} is non-binding, and so
    \begin{align*}
        \Psi'(R_n)
        &= \parfrac{}{R_n}\min_{V_n\in(0,R_n]}\left\{V_n+c\left(\frac{1}{V_n}-\frac{1}{R_n}\right)+\delta\Psi\left(V_n+\sigma^2\right)\right\} \\
        &= \frac{c}{R_n^2}
    \end{align*}
    by the envelope theorem.
    So the optimal posterior variance satisfies the first-order condition
    \begin{align*}
        0
        &= \parfrac{}{V_n}\left\{V_n+c\left(\frac{1}{V_n}-\frac{1}{R_n}\right)+\delta\Psi\left(V_n+\sigma^2\right)\right\} \\
        &= 1-\frac{c}{V_n^2}+\delta\Psi'(V_n+\sigma^2) \\
        &= 1-c\left(\frac{1}{V_n^2}-\frac{\delta}{(V_n+\sigma^2)^2}\right),
    \end{align*}
    which is uniquely solved by~$V$.
    Thus~$\Phi(R_n)=V$ for all~$R_n>R$, from which it follows that~$R=V$ since~$\Phi$ is continuous.
    So the solution to~\eqref{eq:planning-problem} induces residual variances~$R_n,R_{n+1},\ldots$ that satisfy
    \begin{align*}
        R_{i+1}
        &= \Phi(R_i)+\sigma^2 \\
        &= \min\{R_i+\sigma^2,V+\sigma^2\}
    \end{align*}
    for each~$i\ge n$.
    Therefore, if the agent chooses the precisions that solve~\eqref{eq:planning-problem} for each~$n\in\N$, then for each~$n>1$ we have
    \begin{align*}
        R_n
        &= \min\{R_{n-1}+\sigma^2,V+\sigma^2\} \\ 
        &\hspace{0.6em} \vdots \\
        &= \min\{R_1+\sigma^2(n-1)\}\cup\left\{V+\sigma^2i\right\}_{i=1}^{n-1} \\
        &= \min\{\sigma_0^2+\sigma^2t_n,V+\sigma^2\}.
    \end{align*}
    It follows that, for each~$n\in\N$, the optimal precision
    \begin{align*}
        \frac{1}{\min\{R_n,V\}}-\frac{1}{R_n}
        &= \frac{1}{\min\{\sigma_0^2+\sigma^2t_n,V\}}-\frac{1}{\min\{\sigma_0^2+\sigma^2t_n,V+\sigma^2\}} \\
        &= \begin{cases}
            0 & \text{if}\ \sigma_0^2+\sigma^2t_n<V \\
            \frac{1}{V}-\frac{1}{\sigma_0^2+\sigma^2t_n} & \text{if}\ V\le\sigma_0^2+\sigma^2t_n<V+\sigma^2 \\
            \frac{1}{V}-\frac{1}{V+\sigma^2} & \text{if}\ V+\sigma^2\le\sigma_0^2+\sigma^2t_{n-1}
        \end{cases} \\
        &= p_n^\dag
    \end{align*}
    induces posterior variance
    \begin{align*}
        \Var(\theta(t_n)\mid\Hcal_n)
        &= \min\{R_n,V\} \\
        &= \min\{\sigma_0^2+\sigma^2t_n,V\}.\qedhere
    \end{align*}
\end{proof}
\fi

The ``steady-state'' posterior variance~$V$ converges to~$\sqrt{c}$ from below as~$\delta\to0$.
In this limit, the forward-looking precision~$p_n^\dag$ equals the myopic precision~$p_n^*$.
But if~$\delta>0$ then~$V<\sqrt{c}$, and so~$p_n^\dag\ge p_n^*$ with equality if and only if~$p_n^\dag=0$.
Indeed, differentiating~\eqref{eq:planning-V} with respect to~$\delta$ gives
\[ 0=-2\left(\frac{1}{V^3}-\frac{\delta}{(V+\sigma^2)^3}\right)\parfrac{V}{\delta}-\frac{1}{(V+\sigma^2)^2}, \]
from which it follows that~$\partial V/\partial\delta<0$.
Then
\begin{align*}
    \parfrac{p_n^\dag}{\delta}
    &= -\parfrac{V}{\delta}\begin{cases}
        0 & \text{if}\ t_n\le\overline{t} \\
        \frac{1}{V^2} & \text{if}\ n=1\ \text{and}\ \overline{t}<t_1 \\
        \frac{1}{V^2} & \text{if}\ n>1\ \text{and}\ \overline{t}<t_n\le\overline{t}+1 \\
        \frac{1}{V^2}-\frac{1}{(V+\sigma^2)^2} & \text{if}\ n>1\ \text{and}\ \overline{t}+1<t_n
    \end{cases} \\
    &\ge 0
\end{align*}
with equality if and only if~$t_n\le\overline{t}$.
So increasing~$\delta$ lowers the steady-state posterior variance~$V$, making the agent better off, but raises the steady-state optimal precision~$p_n^\dag$, making him worse off.
The former effect dominates: differentiating the steady-state payoff
\[ \lim_{n\to\infty}\psi_n(p_n^\dag)=V+c\left(\frac{1}{V}-\frac{1}{V+\sigma^2}\right)\]
with respect to~$\delta$ gives
\begin{align*}
    \parfrac{}{\delta}\lim_{n\to\infty}\psi_n(p_n^\dag)
    &= \left[1-c\left(\frac{1}{V^2}-\frac{1}{(V+\sigma^2)^2}\right)\right]\parfrac{V}{\delta} \\
    &\overset{\star}{=} \frac{c(1-\delta)}{(V+\sigma^2)^2}\parfrac{V}{\delta} \\
    &< 0,
\end{align*}
where~$\star$ holds by~\eqref{eq:planning-V}.
Thus, in the long run, the agent is better off at each time~$t_n$ when he cares more about his future selves.
This is because the cost of buying more information is smaller than the benefit of having more information bought by his \emph{past} selves.
Thus, internalizing the future value of information makes the agent better off because he can ``share the load'' of staying up-to-date with his past selves.

\section{Other state-generating processes}
\label{sec:other-processes}

This section extends my analysis from Brownian motions to other Gaussian processes.
Such processes can be uniquely defined by
\begin{enumerate}

    \item[(i)]
    a function mapping times~$t\ge0$ to prior means~$\E[\theta(t)]$, and

    \item[(ii)]
    a function mapping pairs of times~$t\ge0$ and~$t'\ge0$ to covariances~$\Cov(\theta(t),\theta(t'))$.

\end{enumerate}
The Brownian motion described in Section~\ref{sec:model} is the Gaussian process characterized by~$\E[\theta(t)]=\mu t$ and~$\Cov(\theta(t),\theta(t'))=\sigma_0^2+\sigma^2\min\{t,t'\}$.
Different covariance functions lead to different residual variances~$R_n$ and, by Lemma~\ref{lem:solution}, different optimal myopic precisions~$p_n^*$.
This is because~$R_n$ depends on how the signals in~$\Hcal_n\setminus\{s_n\}$ covary with each other and with the state's value~$\theta(t_n)$.
I describe this dependence in Lemma~\ref{lem:residual-variance-arbitrary}.
It defines~$R_n$ as a function of the state's autocovariances and signals' precisions when~$\{\theta(t)\}_{t\ge0}$ follows an arbitrary Gaussian process.

\begin{lemma}
    \label{lem:residual-variance-arbitrary}
    Suppose~$\{\theta(t)\}_{t\ge0}$ follows a Gaussian process.
    Let~$n\in\N$, let~$G_n\in\R^{n\times n}$ be the covariance matrix with~${ij}^\text{th}$ entry
    \[ \left[G_n\right]_{ij}=\Cov(\theta(t_i),\theta(t_j)), \]
    and let~$D_n\in(\R\cup\{\infty\})^{n\times n}$ be the diagonal matrix with~${ii}^\text{th}$ entry
    \[ \left[D_n\right]_{ii}=\begin{cases}
        \infty & \text{if}\ i<n\ \text{and}\ p_i=0 \\
        1/p_i & \text{if}\ i<n\ \text{and}\ p_i>0 \\
        0 & \text{if}\ i=n.
    \end{cases} \]
    Then~$R_1=\Var(\theta(t_1))$.
    Moreover, if~$n>1$ and~$(G_n+D_n)$ is invertible, then
    \[ \frac{1}{R_n}=\left[\left(G_n+D_n\right)^{-1}\right]_{nn}. \]
\end{lemma}
\ifbodyproofs\begin{proof}[Proof of Lemma~\ref{lem:residual-variance-arbitrary}]
    Now~$R_1=\Var(\theta(t_1))$ by definition, so suppose~$n>1$.
    Define~$\Sigma_n$ and~$w_{nt}$ as in Lemma~\ref{lem:posterior}, define~$w_n\equiv w_{nt_n}$, and let~$e_n\in\R^n$ be the~$n^\text{th}$ standard basis vector.
    Then
    \[ \Var(\theta(t_n)\mid\Hcal_n)=\Var(\theta(t_n))-w_n^T\Sigma_n^{-1}w_n \]
    whenever
    \[ \Sigma_n=G_n+D_n+\frac{1}{p_n}e_ne_n^T \]
    is invertible.
    If~$(G_n+D_n)$ is invertible, then the Sherman-Morrison formula \citep{Bartlett-1951-AoMS} yields
    \[ \Sigma_n^{-1}=(G_n+D_n)^{-1}-\frac{\frac{1}{p_n}(G_n+D_n)^{-1}e_ne_n^T(G_n+D_n)^{-1}}{1+\frac{1}{p_n}e_n^T(G_n+D_n)^{-1}e_n} \]
    and hence
    \begin{align*}
        \Var(\theta(t_n)\mid\Hcal_n)
        &= \Var(\theta(t_n))-w_n^T\left((G_n+D_n)^{-1}-\frac{(G_n+D_n)^{-1}e_ne_n^T(G_n+D_n)^{-1}}{p_n+e_n^T(G_n+D_n)^{-1}e_n}\right)w_n \\
        &= \Var(\theta(t_n))-w_n^T(G_n+D_n)^{-1}w_n+\frac{\left(w_n^T(G_n+D_n)^{-1}e_n\right)^2}{p_n+e_n^T(G_n+D_n)^{-1}e_n}.
    \end{align*}
    But~$w_n=(G_n+D_n)e_n$ by the definitions of~$G_n$ and~$D_n$, and so
    \begin{align*}
        w_n^T(G_n+D_n)^{-1}w_n
        &= w_n^Te_n \\
        &= \Var(\theta(t_n))
    \end{align*}
    and similarly~$w_n^T(G_n+D_n)^{-1}e_n=1$.
    Thus
    \begin{align*}
        R_n
        &= \Var(\theta(t_n)\mid\Hcal_n\setminus\{s_n\}) \\
        &= \lim_{p_n\to0}\Var(\theta(t_n)\mid\Hcal_n) \\
        &= \lim_{p_n\to0}\frac{1}{p_n+e_n^T(G_n+D_n)^{-1}e_n} \\
        &= \frac{1}{\left[(G_n+D_n)^{-1}\right]_{nn}}.\qedhere
    \end{align*}
\end{proof}
\fi

Combining Lemmas~\ref{lem:solution} and~\ref{lem:residual-variance-arbitrary} allows me to compute the optimal myopic precisions~$p_n^*$ recursively when~$\{\theta(t)\}_{t\ge0}$ follows \emph{any} Gaussian process.
The remainder of this section discusses two such processes.

\subsection{The Ornstein-Uhlenbeck process}
\label{sec:ou}

Suppose~$\{\theta(t)\}_{t\ge0}$ follows an Ornstein-Uhlenbeck (OU) process: it solves the SDE
\[ \der\theta(t)=-\alpha\theta(t)\der t+\sigma\der W(t), \]
where the parameter~$\alpha>0$ controls how quickly~$\{\theta(t)\}_{t\ge0}$ reverts to its mean, and where~$\sigma>0$ and~$\{W(t)\}_{t\ge0}$ are defined as in Section~\ref{sec:model}.
Choosing~$\alpha=\sigma^2/2\sigma_0^2$ yields~$\E[\theta(t)]=0$ and
\begin{equation}
    \Cov(\theta(t),\theta(t'))=\sigma_0^2\exp\left(-\frac{\sigma^2}{2\sigma_0^2}\abs{t-t'}\right) \label{eq:covariance-ou}
\end{equation}
for all~$t\ge0$ and~$t'\ge0$ \citep[see][p. 358]{Karatzas-Shreve-1988-}.
This choice yields a stationary process with prior variances~$\Var(\theta(t))=\sigma_0^2$ that do not vary with~$t$.
In contrast, if~$\{\theta(t)\}_{t\ge0}$ follows a Brownian motion then it is not stationary because then the prior variances~$\Var(\theta(t))=\sigma_0^2+\sigma^2t$ grow with~$t$.

Proposition~\ref{prop:solution-ou} characterizes the optimal myopic precisions~$p_n^*$ when~$\{\theta(t)\}_{t\ge0}$ follows an OU process.
These precisions are strictly positive if and only if the prior variance~$\sigma_0^2$ exceeds the target posterior variance~$\sqrt{c}$.
The agent either
(i)~never buys information or
(ii)~buys information at every opportunity.
He never waits to buy information because the prior variance never changes.
This contrasts with when~$\{\theta(t)\}_{t\ge0}$ follows a Brownian motion and the agent waits until information is more beneficial than costly to start buying it.

\begin{proposition}
    \label{prop:solution-ou}
    Suppose~$\{\theta(t)\}_{t\ge0}$ follows an OU process with scale~$\sigma>0$, initial value~$\theta_0\sim\Ncal(0,\sigma_0^2)$, and mean-reversion parameter~$\alpha=\sigma^2/2\sigma_0^2$.
    Define
    \begin{equation}
        \rho_n \equiv \exp\left(-\frac{\sigma^2}{2\sigma_0^2}(t_n-t_{n-1})\right) \label{eq:solution-ou-correlation}
    \end{equation}
    for each~$n\in\N\setminus\{1\}$.
    Then
    \[ p_n^*=\begin{cases}
        0 & \text{if}\ \sqrt{c}\ge\sigma_0^2 \\
        \frac{1}{\sqrt{c}}-\frac{1}{\sigma_0^2} & \text{if}\ n=1\ \text{and}\ \sqrt{c}<\sigma_0^2 \\
        \frac{1}{\sqrt{c}}-\frac{1}{\rho_n^2\sqrt{c}+(1-\rho_n^2)\sigma_0^2} & \text{if}\ n>1\ \text{and}\ \sqrt{c}<\sigma_0^2
    \end{cases} \]
    for each~$n\in\N$.
\end{proposition}
\ifbodyproofs\begin{proof}[Proof of Proposition~\ref{prop:solution-ou}]
    We first derive an expression for~$R_n$.
    Now~$R_1=\Var(\theta(t_1))=\sigma_0^2$ by definition, so suppose~$n>1$.
    The OU process is a Gauss-Markov process, so we can use Lemma~\ref{lem:residual-variance-markov} to characterize~$R_n$.
    Define~$\beta_n$ and~$\gamma_n$ as in that lemma.
    Then
    \begin{align*}
        \beta_n^2
        &= \left(\frac{\sigma_0^2\exp\left(-\frac{\sigma^2}{2\sigma_0^2}\abs{t_n-t_{n-1}}\right)}{\sigma_0^2}\right)^2 \\
        &= \exp\left(-\frac{\sigma^2}{\sigma_0^2}(t_n-t_{n-1})\right)
    \end{align*}
    and
    \begin{align*}
        \gamma_n
        &\overset{\star}{=} \Var(\theta(t_n))-\frac{\left(\Cov(\theta(t_n),\theta(t_{n-1}))\right)^2}{\Var(\theta(t_{n-1}))} \\
        &= \sigma_0^2-\frac{\left(\sigma_0^2\exp\left(-\frac{\sigma^2}{2\sigma_0^2}\abs{t_n-t_{n-1}}\right)\right)^2}{\sigma_0^2} \\
        &= \left[1-\exp\left(-\frac{\sigma^2}{\sigma_0^2}(t_n-t_{n-1})\right)\right]\sigma_0^2,
    \end{align*}
    where~$\star$ holds by Lemma~\ref{lem:multivariate-normal-conditional}.
    So
    \begin{align*}
        \prod_{j=0}^{n-2}\beta_{n-j}^2
        &= \prod_{j=0}^{n-2}\exp\left(-\frac{\sigma^2}{\sigma_0^2}(t_n-t_{n-1})\right) \\
        &= \exp\left(-\frac{\sigma^2}{\sigma_0^2}\sum_{j=0}^{n-2}(t_{n-j}-t_{n-j-1})\right) \\
        &= \exp\left(-\frac{\sigma^2}{\sigma_0^2}(t_n-t_1)\right)
    \end{align*}
    and
    \begin{align*}
        \sum_{j=0}^{n-2}\gamma_{n-j}\prod_{k=0}^{j-1}\beta_{n-k}^2
        &= \sum_{j=0}^{n-2}\left[1-\exp\left(-\frac{\sigma^2}{\sigma_0^2}(t_{n-j}-t_{n-j-1})\right)\right]\sigma_0^2\prod_{k=0}^{j-1}\exp\left(-\frac{\sigma^2}{\sigma_0^2}(t_{n-k}-t_{n-k-1})\right) \\
        &= \sigma_0^2\sum_{j=0}^{n-2}\left[1-\exp\left(-\frac{\sigma^2}{\sigma_0^2}(t_{n-j}-t_{n-j-1})\right)\right]\exp\left(-\frac{\sigma^2}{\sigma_0^2}(t_{n}-t_{n-j})\right) \\
        &= \sigma_0^2\sum_{j=0}^{n-2}\left[\exp\left(-\frac{\sigma^2}{\sigma_0^2}(t_{n}-t_{n-j})\right)-\exp\left(-\frac{\sigma^2}{\sigma_0^2}(t_n-t_{n-j-1})\right)\right] \\
        &= \sigma_0^2\left[\exp\left(-\frac{\sigma^2}{\sigma_0^2}(t_{n}-t_{n})\right)-\exp\left(-\frac{\sigma^2}{\sigma_0^2}(t_n-t_1)\right)\right] \\
        &= \sigma_0^2\left(1-\exp\left(-\frac{\sigma^2}{\sigma_0^2}(t_n-t_1)\right)\right),
    \end{align*}
    and similarly
    \[ \prod_{j=0}^{i-1}\beta_{n-j}^2=\exp\left(-\frac{\sigma^2}{\sigma_0^2}(t_n-t_{n-i})\right) \]
    and
    \[ \sum_{j=0}^{i-1}\gamma_{n-j}\prod_{k=0}^{j-1}\beta_{n-k}^2=\sigma_0^2\left(1-\exp\left(-\frac{\sigma^2}{\sigma_0^2}(t_n-t_{n-i})\right)\right) \]
    for each~~$i<n$.
    Thus, if~$p_i=p_i^*$ for each~$i<n$, then
    \begin{align*}
        R_n
        &= \min\left\{\Var(\theta(t_1))\prod_{j=0}^{n-2}\beta_{n-j}^2+\sum_{j=0}^{n-2}\gamma_{n-j}\prod_{k=0}^{j-1}\beta_{n-k}^2\right\}\cup\left\{\sqrt{c}\prod_{j=0}^{i-1}\beta_{n-j}^2+\sum_{j=0}^{i-1}\gamma_{n-j}\prod_{k=0}^{j-1}\beta_{n-k}^2\right\}_{i=1}^{n-1} \\
        &= \min\left\{\sigma_0^2\right\}\cup\left\{\sqrt{c}\exp\left(-\frac{\sigma^2}{\sigma_0^2}(t_n-t_{n-i})\right)+\sigma_0^2\left(1-\exp\left(-\frac{\sigma^2}{\sigma_0^2}(t_n-t_{n-i})\right)\right)\right\}_{i=1}^{n-1} \\
        &= \min\left\{\sigma_0^2,\rho_n^2\sqrt{c}+(1-\rho_n^2)\sigma_0^2\right\}
    \end{align*}
    by Lemma~\ref{lem:residual-variance-markov} and the definition of~$\rho_n$.
    The result follows from Lemma~\ref{lem:solution}.
\end{proof}\fi

The parameter~\eqref{eq:solution-ou-correlation} equals the Pearson correlation of~$\theta(t_n)$ and~$\theta(t_{n-1})$.
The square of~$\rho_n$ measures how much of the variance of~$\theta(t_n)$ is explained by the variance of~$\theta(t_{n-1})$.
Increasing~$\rho_n^2$ makes the agent more able to forecast~$\theta(t_n)$ using his estimate of~$\theta(t_{n-1})$, leading him to demand less information and lowering his optimal precision~$p_n^*$.

\subsection{A linear process}

If~$\{\theta(t)\}_{t\ge0}$ follows a Brownian motion or an OU process with~$\sigma_0^2>\sqrt{c}$, then the optimal myopic precision~$p_n^*$ remains bounded away from zero even as the number~$\Abs{\Hcal_n\setminus\{s_n\}}=n-1$ of previous signals grows without bound.
This is because the agent remains willing to buy information about shocks to the state between consecutive time steps.
However, if~$\{\theta(t)\}_{t\ge0}$ is linear in~$t$, then~$p_n^*$ can become arbitrarily small as~$n$ becomes arbitrarily large.
To see why, suppose~$\theta(t)$ is linear in~$t$:
\[ \theta(t)=\theta_0+\beta t \]
for some unknown intercept~$\theta_0\in\R$ and slope~$\beta\in\R$.
Assuming~$\theta_0\sim\Ncal(0,\sigma_0^2)$ and~$\beta\sim\Ncal(\mu,\sigma^2)$ are independently distributed under the agent's prior, the state~$\{\theta(t)\}_{t\ge0}$ follows a Gaussian process with prior mean $\E[\theta(t)]=\mu t$ and covariance
\begin{equation}
    \Cov(\theta(t),\theta(t'))=\sigma_0^2+\sigma^2tt' \label{eq:covariance-linear}
\end{equation}
for all~$t\ge0$ and~$t'\ge0$.
This process differs from the Brownian motion only in the specification of~$\Cov(\theta(t),\theta(t'))$, where the product of~$t$ and~$t'$ replaces their minimum.
Yet this replacement leads to substantively different learning patterns.

\begin{figure}[!t]
    \centering
    \includegraphics[width=0.75\linewidth]{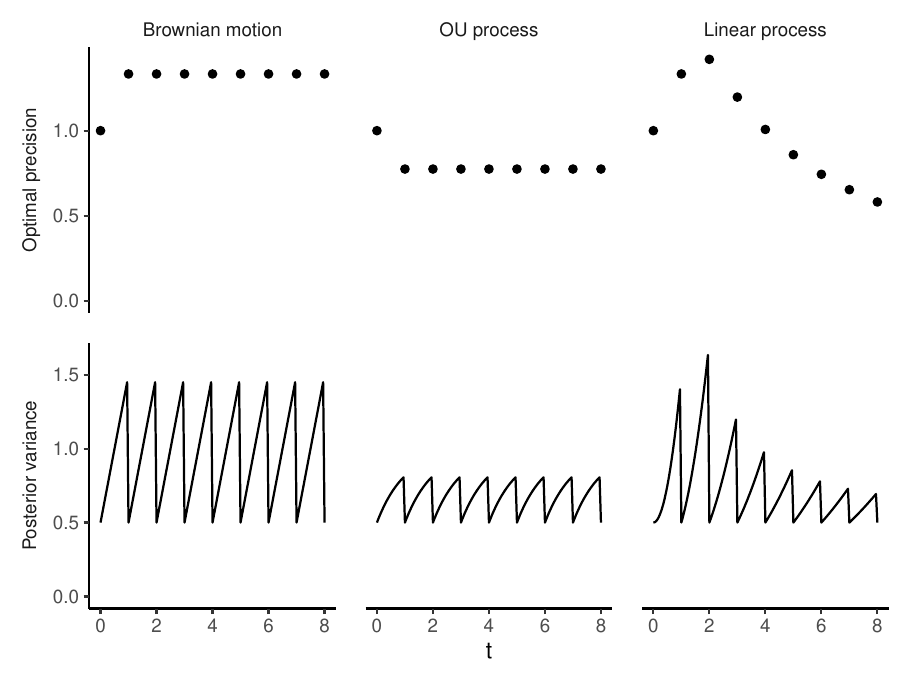}
    \caption{Optimal precisions and posterior variances}
    \label{fig:precisions-variances}
    \caption*{
    \footnotesize
    {\itshape Notes:}
    Points in top row represent optimal myopic precisions~$p_n^*$ computed via Lemmas~\ref{lem:solution} and~\ref{lem:residual-variance-arbitrary}.
    Lines in bottom row represent posterior variances~$\Var(\theta(t)\mid\Hcal_n)$ with~$t\in[t_n,t_{n+1})$ and~$p_n=p_n^*$ for each~$n\in\N$.
    Left/middle/right column corresponds to case~(i)/(ii)/(iii) defined in the text.
    All cases have~$(\sigma_0,\sigma,c)=(1,1,0.25)$ and~$t_n=n-1$ for each~$n\in\N$.
    }
\end{figure}
I illustrate these patterns in Figure~\ref{fig:precisions-variances}.
It shows the optimal myopic precisions~$p_n^*$ and resulting posterior variances when the state follows
\begin{enumerate}

    \item[(i)]
    the Brownian motion defined in Section~\ref{sec:model},

    \item[(ii)]
    the OU process defined in Section~\ref{sec:ou}, and
    
    \item[(iii)]
    the linear process defined above.

\end{enumerate}
I choose~$(\sigma_0,\sigma,c)=(1,1,0.25)$ and~$t_n=n-1$ in all three cases.
This choice makes the first precision~$p_1^*=1$ and the target posterior variance~$\sqrt{c}=0.5$ constant across cases.
But the subsequent precisions~$p_2^*,p_3^*,\ldots$ are not constant across cases.
In case~(i) we have~$R_n=\sqrt{c}+\sigma^2(t_n-t_{n-1})=3/2$ and~$p_n^*=1/\sqrt{c}-1/R_n=4/3$ for each~$n\ge2$, while in case~(ii) we have
\begin{align*}
    R_n
    &= \exp\left(-\frac{\sigma^2}{\sigma_0^2}(t_n-t_{n-1})\right)\sqrt{c}+\left(1-\exp\left(-\frac{\sigma^2}{\sigma_0^2}(t_n-t_{n-1})\right)\right)\sigma_0^2 \\
    &= 1-0.5\exp(-1) \\
    &\approx 0.632
\end{align*}
and~$p_n^*\approx0.775$ for each~$n\ge2$.
These precisions~$p_n^*$ are case-wise constant in~$n$ because the Brownian motion and OU processes are Markov processes, and so the agent faces the same problem~\eqref{eq:problem} at each time~$t_n$.
In contrast, in case~(iii), the optimal myopic precisions~$p_n^*$ generically vary with~$n$ and eventually decrease with~$n$.
Intuitively, each signal with~$p_n>0$ provides information about the initial value~$\theta_0$ and growth rate~$\beta$, decreasing the residual variance
\begin{align*}
    R_n
    &= \Var(\theta_0+\beta t_n\mid\Hcal_n\setminus\{s_n\}) \\
    &= \Var(\theta_0\mid\Hcal_n\setminus\{s_n\})+2t_n\Cov(\theta_0,\beta\mid\Hcal_n\setminus\{s_n\})+t_n^2\Var(\beta\mid\Hcal_n\setminus\{s_n\})
\end{align*}
of~$\theta(t_n)$ and lowering the value of new information.
So as~$n\to\infty$, the agent's posterior on~$(\theta_0,\beta)$ becomes arbitrarily precise and his demand for information about~$\theta(t_n)=\theta_0+\beta t_n$ becomes arbitrarily small.%
\footnote{
This intuition relies on the time steps~$\Delta t_n\equiv t_n-t_{n-1}$ satisfying a regularity condition that ensures the posterior variance~$\Var(\beta\mid\Hcal_n\setminus\{s_n\})$ shrinks faster than~$t_n^2$ grows; otherwise, the residual variance~$R_n$ grows without bound and so~$p_n^*$ remains bounded away from zero as~$n\to\infty$.
I defer identifying this regularity condition for future research.
}

\section{Related literature}
\label{sec:literature}

This paper fits into the literature on dynamic information acquisition.
Many papers in that literature \citep[e.g.,][]{Bolton-Harris-1999-ECTA,Che-Mierendorff-2019-AER,Fudenberg-etal-2018-AER,Hebert-Woodford-2023-JET,Liang-etal-2018-,Liang-etal-2022-ECTA,Moscarini-Smith-2001-ECTA} study the optimal strategy for learning about a fixed state before taking an action at a (possibly endogenous) time.
In contrast, this paper studies the optimal strategy for learning about a changing state before taking actions at many times.
\cite{Immorlica-etal-2021-LeibnizInt.Proc.Inform.LIPIcs} also study such a strategy in a framework similar to mine.
But they focus on allocating a ``precision budget'' across possible action times, whereas I focus on the tension between the cost and value of information at every action time.

This focus is common in papers on rational inattention \citep{Caplin-Dean-2015-AER,Mackowiak-etal-2023-JEL,Sims-2003-JME}.
The closest such paper is by \cite{Steiner-etal-2017-ECTA}.
They consider an agent who observes signals of a time-varying state before taking actions at discrete times.
\citeauthor{Steiner-etal-2017-ECTA} assume the cost of each signal is proportional to the expected entropy reduction it delivers.
This cost function is prior-dependent, which complicates \citeauthor{Steiner-etal-2017-ECTA}'s dynamic analysis because each signal changes future beliefs and, thus, future costs.
In contrast, I assume the cost of each signal is proportional to its precision.
This cost function is prior-independent, allowing me to solve the agent's dynamic problem analytically (see Section~\ref{sec:planning}).

\cite{Barilla-2025-} also studies dynamic information acquisition with a time-varying state.
He considers an agent who buys signals of a binary state at endogenous times.
In contrast, I consider an agent who buys signals of a continuous state at exogenous times.
Both \cite{Barilla-2025-} and \cite{Steiner-etal-2017-ECTA} assume actions have arbitrary payoffs, whereas I assume they have quadratic payoffs.

My ``Gaussian-quadratic'' setup is decidedly tractable, which explains its use in many related papers.
Three such papers are by \cite{Colombo-etal-2014-REStud}, \cite{Hellwig-Veldkamp-2009-REStud}, and \cite{Myatt-Wallace-2012-REStud}.
They study agents' choices of how much to learn about a payoff-relevant state and the externalities these choices impose on other agents.
Likewise, the agent in my model makes choices that impose externalities on his future selves.
However, I allow the payoff-relevant state to vary across these selves.
Indeed, my model could also be used to study a setting where different agents want to learn different states.
I assume these states covary in a single dimension---time---but one could extend my analysis to states that covary in many dimensions.

This many-agent interpretation of my analysis connects it to the literature on social learning.
A close paper in that literature is by \cite{Dasaratha-etal-2023-REStud}.
They consider a sequence of agents who observe private signals and others' actions before estimating a time-varying state.
They assume signals are costless and that agents have finite memories.
In contrast, I assume signals are costly and that the agent has an infinite memory.

\section{Conclusion}
\label{sec:conclusion}

This paper studies optimal learning about a time-varying state.
I assume the state follows a Gaussian process, and that a long-lived Bayesian agent receives Gaussian signals and takes actions with quadratic state-dependent costs.
These assumptions make my model tractable and yield simple expressions for the optimal signal precisions.
I discuss how these precisions depend on whether the agent is myopic or forward-looking, and on process driving the state's evolution.

Natural extensions of my analysis are to non-Gaussian processes, multi-dimensional processes, and non-quadratic payoffs.
One could also consider an agent with mis-specified priors or limited memory.
I defer these extensions to future research.

{
\raggedright
\bibliographystyle{apalike}
\bibliography{references}
}

\appendix

\clearpage
\section{Proofs}
\label{sec:proofs}

\counterwithin{lemma}{section}
\setcounter{lemma}{0}

\subsection{Ancillary results}

\begin{lemma}
    \label{lem:multivariate-normal-conditional}
    Let~$n\in\N$, and let~$z\sim\Ncal(\mu,\Sigma)$ be multivariate normally distributed with mean~$\mu\in\R^n$ and covariance matrix~$\Sigma\in(\R\cup\{\infty\})^{n\times n}$.
    Let~$z=(z_1,z_2)$ be a partition of~$z$ into vectors~$z_1\in\R^m$ and~$z_2\in\R^{n-m}$ with~$1\le m<n$, and let~$\mu=(\mu_1,\mu_2)$ and
    \[ \Sigma=\begin{pmatrix}
        \Sigma_{11} & \Sigma_{12} \\
        \Sigma_{21} & \Sigma_{22}
    \end{pmatrix} \]
    be the corresponding partitions of~$\mu$ and~$\Sigma$.
    If~$\Sigma_{22}$ is invertible, then
    \[ z_1\mid z_2\sim\Ncal\left(\mu_1+\Sigma_{12}\Sigma_{22}^{-1}(z_2-\mu_2),\ \Sigma_{11}-\Sigma_{12}\Sigma_{22}^{-1}\Sigma_{21}\right). \]
\end{lemma}
\begin{proof}\let\qed\relax
    See, e.g., \citet[p.\! 87]{Bishop-2006-} or \citet[p.\! 55]{DeGroot-2004-}.
\end{proof}

\begin{lemma}
    \label{lem:posterior}
    Let~$t\ge0$ and~$n\in\N$.
    Let~$y_n=(s_1,s_2,\ldots,s_n)\in\R^n$ be the (random) vector of signals in~$\Hcal_n$, let~$\Sigma_n\in(\R\cup\{\infty\})^{n\times n}$ be the covariance matrix with~${ij}^\text{th}$ entry
    \begin{align*}
        \left[\Sigma_n\right]_{ij}
        &= \Cov(s_i,s_j) \\
        &= \begin{cases}
            \infty & \text{if}\ i=j\ \text{and}\ p_i=0 \\
            \Var(\theta(t_i))+1/p_i & \text{if}\ i=j\ \text{and}\ p_i>0 \\
            \Cov(\theta(t_i),\theta(t_j)) & \text{otherwise},
        \end{cases}
    \end{align*}
    and let~$w_{nt}\in\R^n$ be the vector with~$i^\text{th}$ component~$[w_{nt}]_i=\Cov(\theta(t),\theta(t_i))$.
    If~$\Sigma_n$ is invertible, then
    \[ \theta(t)\mid\Hcal_n\sim\Ncal\left(\E[\theta(t)]+w_{nt}^T\Sigma_n^{-1}\left(y_n-\E[y_n]\right),\Var(\theta(t))-w_{nt}^T\Sigma_n^{-1}w_{nt}\right). \]
\end{lemma}
\begin{proof}
    Consider the vector~$z=(y_n,\theta(t))\in\R^{n+1}$ obtained by appending~$\theta(t)$ to~$y_n$.
    This vector is multivariate normally distributed with mean~$\E[z]=(\E[y_n],\E[\theta(t)])$ and covariance matrix
    \[ \Var(z)=\begin{bmatrix}
        \Var(\theta(t)) & w_{nt}^T \\
        w_{nt}^T & \Sigma_n
    \end{bmatrix}. \]
    The result follows immediately from Lemma~\ref{lem:multivariate-normal-conditional}.
\end{proof}

\begin{lemma}
    \label{lem:residual-variance-markov}
    Suppose~$\{\theta(t)\}_{t\ge0}$ follows a Gauss-Markov process.
    Define
    \[ \beta_n\equiv\frac{\Cov(\theta(t_n),\theta(t_{n-1}))}{\Var(\theta(t_{n-1}))} \]
    and
    \[ \gamma_n\equiv\Var(\theta(t_n)\mid\theta(t_{n-1})) \]
    for each~$n\in\N\setminus\{1\}$.
    Then
    \begin{equation}
        R_n=\begin{cases}
            \Var(\theta(t_1)) & \text{if}\ n=1 \\
            \beta_n^2\Var(\theta(t_{n-1})\mid\Hcal_{n-1})+\gamma_n & \text{if}\ n>1.
        \end{cases} \label{eq:residual-variance-markov}
    \end{equation}
    Moreover, if~$n>1$ and~$p_i=p_i^*$ for each~$i<n$, then
    \begin{equation}
        R_n=\min\left\{\Var(\theta(t_1))\prod_{j=0}^{n-2}\beta_{n-j}^2+\sum_{j=0}^{n-2}\gamma_{n-j}\prod_{k=0}^{j-1}\beta_{n-k}^2\right\}\cup\left\{\sqrt{c}\prod_{j=0}^{i-1}\beta_{n-j}^2+\sum_{j=0}^{i-1}\gamma_{n-j}\prod_{k=0}^{j-1}\beta_{n-k}^2\right\}_{i=1}^{n-1}. \label{eq:residual-variance-markov-closed}
    \end{equation}
\end{lemma}
\begin{proof}
    Now~$\Hcal_1\setminus\{s_1\}$ is empty and so~$R_1=\Var(\theta(t_1))$ by definition.

    Suppose~$n>1$.
    If~$\{\theta(t)\}_{t\ge0}$ follows a Markov process, then~$\theta(t_n)$ is conditionally independent of~$\Hcal_{n-1}$ given~$\theta(t_{n-1})$, and so
    \begin{align*}
        \Var(\theta(t_n)\mid\Hcal_{n-1})
        &= \Var(\E[\theta(t_n)\mid\Hcal_{n-1},\theta(t_{n-1})]\mid\Hcal_{n-1})+\E[\Var(\theta(t_n)\mid\Hcal_{n-1},\theta(t_{n-1}))\mid\Hcal_{n-1}] \\
        &= \Var(\E[\theta(t_n)\mid\theta(t_{n-1})]\mid\Hcal_{n-1})+\E[\Var(\theta(t_n)\mid\theta(t_{n-1}))\mid\Hcal_{n-1}]
    \end{align*}
    by the law of total variance.
    But if~$\{\theta(t)\}_{t\ge0}$ follows a Gaussian process, then~$\Var(\theta(t_n)\mid\theta(t_{n-1}))$ is non-stochastic and
    \[ \E[\theta(t_n)\mid\theta(t_{n-1})]=\E[\theta(t_n)]+\frac{\Cov(\theta(t_n),\theta(t_{n-1}))}{\Var(\theta(t_{n-1}))}(\theta(t_{n-1})-\E[\theta(t_{n-1})]) \]
    by Lemma~\ref{lem:multivariate-normal-conditional}.
    Thus
    \begin{align*}
        \Var(\theta(t_n)\mid\Hcal_{n-1})
        &= \Var\left(\E[\theta(t_n)]+\frac{\Cov(\theta(t_n),\theta(t_{n-1}))}{\Var(\theta(t_{n-1}))}(\theta(t_{n-1})-\E[\theta(t_{n-1})])\mid\Hcal_{n-1}\right) \\
        &\quad +\Var(\theta(t_n)\mid\theta(t_{n-1})) \\
        &= \left(\frac{\Cov(\theta(t_n),\theta(t_{n-1}))}{\Var(\theta(t_{n-1}))}\right)^2\Var(\theta(t_{n-1})\mid\Hcal_{n-1})+\Var(\theta(t_n)\mid\theta(t_{n-1}))
    \end{align*}
    since~$\E[\theta(t_n)]$, $\Cov(\theta(t_n),\theta(t_{n-1}))$, $\Var(\theta(t_{n-1}))$, and~$\E[\theta(t_{n-1})]$ are non-stochastic.
    Equation~\eqref{eq:residual-variance-markov} follows.

    Now suppose~$n>1$ and~$p_i=p_i^*$ for each~$i<n$.
    Then, by Lemma~\ref{lem:solution}, we have
    \begin{align*}
        R_n
        &= \beta_n^2\min\{R_{n-1},\sqrt{c}\}+\gamma_n \\
        &= \min\{\beta_n^2R_{n-1}+\gamma_n,\beta_n^2\sqrt{c}+\gamma_n\}.
    \end{align*}
    Substituting this expression into itself recursively gives
    \begin{align*}
        R_n
        &= \min\{\beta_n^2\min\{\beta_{n-1}^2R_{n-2}+\gamma_{n-1},\beta_{n-1}^2\sqrt{c}+\gamma_{n-1}\}+\gamma_n,\beta_n^2\sqrt{c}+\gamma_n\} \\
        &= \min\{\beta_n^2\beta_{n-1}^2R_{n-2}+\beta_n^2\gamma_{n-1}+\gamma_n, \\
        &\hspace{3.7em} \beta_n^2\beta_{n-1}^2\sqrt{c}+\beta_n^2\gamma_{n-1}+\gamma_n, \\
        &\hspace{3.7em} \beta_n^2\sqrt{c}+\gamma_n\} \\
        &= \min\{\beta_n^2\beta_{n-1}^2\beta_{n-2}^2R_{n-3}+\beta_n^2\beta_{n-1}^2\gamma_{n-2}+\beta_n^2\gamma_{n-1}+\gamma_n, \\
        &\hspace{3.7em} \beta_n^2\beta_{n-1}^2\beta_{n-2}^2\sqrt{c}+\beta_n^2\beta_{n-1}^2\gamma_{n-2}+\beta_n^2\gamma_{n-1}+\gamma_k, \\
        &\hspace{3.7em} \beta_n^2\beta_{n-1}^2\sqrt{c}+\beta_n^2\gamma_{n-1}+\gamma_n, \\
        &\hspace{3.7em} \beta_n^2\sqrt{c}+\gamma_n\} \\
        &\hspace{0.6em} \vdots \\
        &= \min\left\{R_1\prod_{j=0}^{n-2}\beta_{n-j}^2+\sum_{j=0}^{n-2}\gamma_{n-j}\prod_{k=0}^{j-1}\beta_{n-k}^2\right\}\cup\left\{\sqrt{c}\prod_{j=0}^{i-1}\beta_{n-j}^2+\sum_{j=0}^{i-1}\gamma_{n-j}\prod_{k=0}^{j-1}\beta_{n-k}^2\right\}_{i=1}^{n-1}.\qedhere
    \end{align*}
\end{proof}

\subsection{Main results}

\ifbodyproofs\else\fi

\ifbodyproofs\else\begin{proof}[Proof of Lemma~\ref{lem:residual-variance-brownian}]
    Now~$R_1=\Var(\theta(t_1))=\sigma_0^2+\sigma^2t_1$ by definition.
    
    Suppose~$n>1$.
    The Brownian motion is a Gauss-Markov process, so we can use Lemma~\ref{lem:residual-variance-markov} to characterize~$R_n$.
    Define~$\beta_n$ and~$\gamma_n$ as in that lemma.
    Then
    \begin{align*}
        \beta_n
        &= \frac{\sigma_0^2+\sigma^2\min\{t_n,t_{n-1}\}}{\sigma_0^2+\sigma^2t_{n-1}} \\
        &= 1
    \end{align*}
    and
    \begin{align*}
        \gamma_n
        &\overset{\star}{=} \Var(\theta(t_n))-\frac{\left(\Cov(\theta(t_n),\theta(t_{n-1}))\right)^2}{\Var(\theta(t_{n-1}))} \\
        &= \sigma_0^2+\sigma^2t_n-\frac{\left(\sigma_0^2+\sigma^2\min\{t_n,t_{n-1}\}\right)^2}{\sigma_0^2+\sigma^2t_{n-1}} \\
        &= \sigma^2(t_n-t_{n-1}),
    \end{align*}
    where~$\star$ holds by Lemma~\ref{lem:multivariate-normal-conditional}.
    The result follows from Lemma~\ref{lem:residual-variance-markov}.
\end{proof}
\fi

\ifbodyproofs\else\begin{proof}[Proof of Proposition~\ref{prop:solution-brownian}]
    Suppose~$n>1$.
    The Brownian motion is a Gauss-Markov process, so we can use Lemma~\ref{lem:residual-variance-markov} to characterize~$R_n$.
    Define~$\beta_n$ and~$\gamma_n$ as in that lemma.
    Then~$\beta_n=1$ and~$\gamma_n=\sigma^2(t_n-t_{n-1})$ from the proof of Lemma~\ref{lem:residual-variance-brownian}.
    So
    \[ \prod_{j=0}^{n-2}\beta_{n-j}^2=1 \]
    and
    \begin{align*}
        \sum_{j=0}^{n-2}\gamma_{n-j}\prod_{k=0}^{j-1}\beta_{n-k}^2
        &= \sum_{j=0}^{n-2}\sigma^2(t_{n-j}-t_{n-j-1}) \\
        &= \sigma^2(t_n-t_1),
    \end{align*}
    and similarly
    \[ \prod_{j=0}^{i-1}\beta_{n-j}^2=1 \]
    and
    \[ \sum_{j=0}^{i-1}\gamma_{n-j}\prod_{k=0}^{j-1}\beta_{n-k}^2=\sigma^2(t_n-t_{n-i}) \]
    for each~~$i<n$.
    Thus, if~$p_i=p_i^*$ for each~$i<n$, then
    \begin{align*}
        R_n
        &= \min\left\{\Var(\theta(t_1))\prod_{j=0}^{n-2}\beta_{n-j}^2+\sum_{j=0}^{n-2}\gamma_{n-j}\prod_{k=0}^{j-1}\beta_{n-k}^2\right\}\cup\left\{\sqrt{c}\prod_{j=0}^{i-1}\beta_{n-j}^2+\sum_{j=0}^{i-1}\gamma_{n-j}\prod_{k=0}^{j-1}\beta_{n-k}^2\right\}_{i=1}^{n-1} \\
        &= \min\left\{\left(\sigma_0^2+\sigma^2t_1\right)+\sigma^2(t_n-t_1)\right\}\cup\left\{\sqrt{c}+\sigma^2(t_n-t_{n-i})\right\}_{i=1}^{n-1} \\
        &= \min\left\{\sigma_0^2+\sigma^2t_n,\sqrt{c}+\sigma^2(t_n-t_{n-1})\right\}
    \end{align*}
    by Lemma~\ref{lem:residual-variance-markov}.
    The result follows from Lemma~\ref{lem:solution} and the definition of~$\overline{t}$.
\end{proof}
\fi

\ifbodyproofs\else\fi

\ifbodyproofs\else\begin{proof}[Proof of Lemma~\ref{lem:residual-variance-arbitrary}]
    Now~$R_1=\Var(\theta(t_1))$ by definition, so suppose~$n>1$.
    Define~$\Sigma_n$ and~$w_{nt}$ as in Lemma~\ref{lem:posterior}, define~$w_n\equiv w_{nt_n}$, and let~$e_n\in\R^n$ be the~$n^\text{th}$ standard basis vector.
    Then
    \[ \Var(\theta(t_n)\mid\Hcal_n)=\Var(\theta(t_n))-w_n^T\Sigma_n^{-1}w_n \]
    whenever
    \[ \Sigma_n=G_n+D_n+\frac{1}{p_n}e_ne_n^T \]
    is invertible.
    If~$(G_n+D_n)$ is invertible, then the Sherman-Morrison formula \citep{Bartlett-1951-AoMS} yields
    \[ \Sigma_n^{-1}=(G_n+D_n)^{-1}-\frac{\frac{1}{p_n}(G_n+D_n)^{-1}e_ne_n^T(G_n+D_n)^{-1}}{1+\frac{1}{p_n}e_n^T(G_n+D_n)^{-1}e_n} \]
    and hence
    \begin{align*}
        \Var(\theta(t_n)\mid\Hcal_n)
        &= \Var(\theta(t_n))-w_n^T\left((G_n+D_n)^{-1}-\frac{(G_n+D_n)^{-1}e_ne_n^T(G_n+D_n)^{-1}}{p_n+e_n^T(G_n+D_n)^{-1}e_n}\right)w_n \\
        &= \Var(\theta(t_n))-w_n^T(G_n+D_n)^{-1}w_n+\frac{\left(w_n^T(G_n+D_n)^{-1}e_n\right)^2}{p_n+e_n^T(G_n+D_n)^{-1}e_n}.
    \end{align*}
    But~$w_n=(G_n+D_n)e_n$ by the definitions of~$G_n$ and~$D_n$, and so
    \begin{align*}
        w_n^T(G_n+D_n)^{-1}w_n
        &= w_n^Te_n \\
        &= \Var(\theta(t_n))
    \end{align*}
    and similarly~$w_n^T(G_n+D_n)^{-1}e_n=1$.
    Thus
    \begin{align*}
        R_n
        &= \Var(\theta(t_n)\mid\Hcal_n\setminus\{s_n\}) \\
        &= \lim_{p_n\to0}\Var(\theta(t_n)\mid\Hcal_n) \\
        &= \lim_{p_n\to0}\frac{1}{p_n+e_n^T(G_n+D_n)^{-1}e_n} \\
        &= \frac{1}{\left[(G_n+D_n)^{-1}\right]_{nn}}.\qedhere
    \end{align*}
\end{proof}
\fi

\ifbodyproofs\else\fi

\end{document}